\documentclass[11pt]{article}

\usepackage{times}
\usepackage{epsfig}
\usepackage{color}
\usepackage{amssymb,amsmath}
\usepackage{verbatim}
\usepackage{color}
\usepackage{algorithmic}
\usepackage{hyperref}
\usepackage{amsthm}

\newtheorem{theorem}{Theorem}[section]
\newtheorem{lemma}[theorem]{Lemma}

\renewcommand{\geq}{\geqslant}
\renewcommand{\leq}{\leqslant}

\newcommand{\Ind}{\mathsf{AI}}
\newcommand{\ind}{\mathsf{f_{AI}}}
\newcommand{\Indk}{\Ind^{\rightarrow k}(k,N)}

\renewcommand{\Pr}{\mathsf{Pr}}
\newcommand{\eps}{\varepsilon}

\renewcommand{\P}{\mathcal{P}}

\newcommand{\integer}  {\mathbb{N}}

\newcommand{\cancel}[1]{}

\renewcommand{\paragraph}[1]{\medskip \noindent{\bf #1.}}

\date{}

\title{Tracking the Frequency Moments at All Times}
\author{Zengfeng Huang \and Wai Ming Tai \and Ke Yi}

\begin{document}
\maketitle

\begin{abstract}
  The traditional requirement for a randomized streaming algorithm is just {\em
    one-shot}, i.e., algorithm should be correct (within the stated $\eps$-error bound)
  at the end of the stream.  In this paper, we study the {\em tracking} problem, where
  the output should be correct at all times.  The standard approach for solving the
  tracking problem is to run $O(\log m)$ independent instances of the one-shot
  algorithm and apply the union bound to all $m$ time instances.  In this paper, we
  study if this standard approach can be improved, for the classical frequency moment
  problem.  We show that for the $F_p$ problem for any $1 < p \le 2$, we actually only
  need $O(\log \log m + \log n)$ copies to achieve the tracking guarantee in the cash
  register model, where $n$ is the universe size.  Meanwhile, we present a lower bound
  of $\Omega(\log m \log\log m)$ bits for all linear sketches achieving this guarantee.
  This shows that our upper bound is tight when $n=(\log m)^{O(1)}$.  We also present
  an $\Omega(\log^2 m)$ lower bound in the turnstile model, showing that the standard
  approach by using the union bound is essentially optimal.
\end{abstract}

\section{Introduction}
All classical randomized streaming algorithms provide a {\em one-shot} probabilistic
guarantee, i.e., the output of the algorithm at the end of the stream is within the
stated $\eps$-error bound with a constant probability.  In many practical applications
where one wants to monitor the status of the stream continuously as it evolves over
time, such a one-shot guarantee is too weak.  Instead, a stronger guarantee, which
requires that the algorithm be correct at {\em all} times, would be desired.  We refer
to this stronger guarantee the {\em tracking} problem.  The standard approach for
solving the tracking problem is to simply reduce the failure probability of the
one-shot algorithm to $O(1/m)$, where $m$ is the length of the stream.  This can be
achieved by running $O(\log m)$ independent instances of the algorithm and returning
the median.  Then by the union bound, with at least constant probability, the output is
correct (i.e., within the stated $\eps$-error bound) at all times.  However, the union
bound may be far from being tight as the $m$ time instances are highly correlated.
Thus, the question we ask in this paper is: Can this $O(\log m)$ factor be further
improved?

We consider this question with the classical frequency moments problem, which is one of
the most extensively studied problems in the streaming literature.  Let $S=(a_1, a_2,
..., a_m)$ be a stream of items, where $a_i\in[n]$ for all $i$. Let $f =
(f_1,\dots,f_n)$ denote the frequency vector of $S$, i.e., $f_i=|\{j:a_j=i\}|$ is the
number of occurrences of $i$ in the stream $S$.  The {\em $p$-th frequency moment} of
$f$ is
$$F_p(f)=\sum_{i=1}^n f_i^p.$$
In particular, $F_1 = m$ and $F_0$ is the number of
distinct items in $S$.  This model is also known as the {\em cash register model}.  In
the related turnstile model, we also allow deletion of items, i.e., each element in the
stream is a pair $(a_i,u_i)$, where $a_i\in[n]$ and $u_i\in\{-1,+1\}$.  The frequency
vector is then defined as $f_i=|\sum_{j:a_j=i}u_j|$.  

\paragraph{Our results}
In Section~\ref{sec:tracking-problem-f_2} we consider the $F_2$ tracking problem.  The
classical AMS sketch \cite{alon99} gives a one-shot estimate to the $F_2$ with $\eps$
relative error with constant probability.  In the turnstile model, it uses $O(\log m +
\log\log n)$ bits of space, which is optimal \cite{kane10:_exact}.  (For simplicity of
presentation, we suppress the dependency on $\eps$ in stating the bounds.)  In the cash
register model, it is also possible to implement the sketch with $O(\log\log m +\log
n)$ bits \cite{alon99} using probabilistic counting \cite{flajolet85:_probab}, so the
space needed is $O(\min\{\log m + \log\log n, \log\log m + \log n\})$, which is also
optimal\footnote{An $\Omega(\min\{\log m, \log n\})$ lower bound is shown in
  \cite{alon99}; the $\Omega(\log\log m)$ lower bound holds trivially since the output
  has at least so many bits if it is a constant-approximation of $F_2$; an
  $\Omega(\log\log n)$ lower bound is shown for the turnstile model in
  \cite{kane10:_exact}, but it actually also holds for the cash register model for any
  small constant $\eps$.}.  Directly using the union bound for the tracking problem
would need $O(\log m)$ independent copies of the AMS sketch, but we show that in the
cash register model, only $O(\log\log m + \log n)$ copies are actually needed. The
$\log n$ factor can be replaced by $\log F_0$, so this bound is never worse than that
obtained by the union bound since $F_0\le n$, and can be much smaller when $m \gg n$.

We also provide lower bounds for the $F_2$ tracking problem, though our lower bounds
require that the sketch has to be linear, i.e., it can be written as $Af$ where $A$ is
some random matrix and $f$ is the frequency vector.  In the cash register model, we
show that any linear sketch for the $F_2$ tracking problem must use $\Omega(\log m
\log\log m)$ bits.  As the $O(\log\log m + \log n)$-bit implementation of the AMS
sketch uses probabilistic counting, it is no longer a linear sketch, so our upper bound
for the $F_2$ tracking problem when restricted to a linear sketch is $O((\log m + \log
\log n)(\log\log m + \log n))$, which matches the lower bound when $n=(\log m)^{O(1)}$.
For non-linear sketches, the upper bound can be $O((\log\log m)^2)$, so the same lower
bound cannot hold, but we currently do not have a lower bound for non-linear sketches.
For the turnstile model, we show a lower bound of $\Omega(\log^2 m)$ bits.  This means
that the standard solution of running $O(\log m)$ copies of the AMS sketch and applying
the union bound is already optimal.

Our upper bound analysis extends to any $F_p, 1< p \le 2$, while our lower bounds hold
for any $F_p, 0<p\le 2$.
 
\section{Tracking problem of $F_2$}
\label{sec:tracking-problem-f_2}
The well-known (fast) AMS sketch \cite{alon99, thorup2004:_tabulation} can be used to obtain a {\em one-shot} estimate of the
$F_2$ with constant probability.  It uses two hash functions: a 4-wise independent hash
function $g:[n]\rightarrow\{+1,-1\}$ and a pairwise independent hash function
$h:[n]\rightarrow[k]$.  Given a frequency vector $f = (f_1,\dots, f_n)$ of some $S$, it
computes $k$ counters $c_j = \sum_{i\in [n], h(i) = j} f_i g(i), j=1,\dots,k$,
and returns $\hat{X} = \sum_{j=1}^k c_j^2$ as the estimate of $F_2(S)$.  It has
been shown that for $k=O(1/\eps^2)$, the AMS sketch returns an $\eps$-approximation of
$F_2(S)$ with constant probability.  The success probability can be boosted to
$1-\delta$ by maintaining $O(\log(1/\delta))$ independent copies of the sketch and
returning the median.  To solve the tracking problem, one could pick $\delta =
\Theta(1/m)$ and apply the union bound, which implies that $O(\log m)$ copies would be
needed.  Below, we give a tighter analysis showing that only $O(\log F_0 + \log\log m +
\log(1/\eps))$ copies are actually needed, where $F_0$ is the number of distinct
elements in $S$.

\begin{theorem}
\label{thm:main}
  Given a stream $S=(a_1, a_2, ..., a_m)$ where $a_i\in [n]$, let
  $S_i=(a_1,\dots,a_i)$.  The stream is fed to $O(\log F_0+\log \log m + \log(1/\eps))$
  independent copies of the AMS sketch, where $F_0$ is the number of distinct elements
  in $S$ and $\eps>0$ is any small positive real.  Let $\hat{X}_i$ be the median
  estimate of the sketches after processing $S_i$, then
  $\Pr\left(\bigwedge_{i=1}^{m}|\hat{X_i}-F_2(S_i)|<\epsilon F_2(S_i)\right)>1/2$.
\end{theorem}

We will consider every frequency vector as a $n$-dimensional point.  The basic idea of
the proof is thus to show that nearby points are highly correlated: If the AMS sketch
produces an accurate estimate at one point $a$, then with good probability it is also
accurate at all points within a ball centered at $a$.  More precisely, we view every frequency vector $f$ lying on the $n$-dimensional Euclidean space $\mathbb{R}^{n}$.  For a frequency vectors $x =
(x_1,\dots, x_n) \in \mathbb{R}^{n}$, the approximation ratio of the AMS sketch using
hash functions $g$ and $h$ is
\[ F_{g,h}(x)=(\sum_{j=1}^{k}\left(\sum_{i=1}^{n}g(i)I(h(i)=j)x_i\right)^2)/(x^tx)=x^tHx/x^tx,\]
where $H_{i,j}=g(i)g(j)I(h(i)=h(j))$.

We use $F(x)$ to denote the random variable $F_{g,h}(x)$ when $g,h$ are randomly
chosen.  For any $a\in \mathbb{R}^n$ and $r>0$, denote by $B(a, r)$ the ball centered
at $a$ with radius $r$ (using 1-norm distance).  Let $T_0$ be the set of distinct
elements appearing in $S$; note that $|T_0| = F_0$.  Denote by $P$ the subspace of
$\mathbb{R}^n$ spanned by the elements of $T_0$, i.e., $P = \{ x=(x_1,\dots, x_n) \mid
x_i \in \mathbb{R} \text{ if } i\in S_0, \text{else } x_i = 0\}$. For $j=1,\dots, k$, let $T_j=\{i\in T_0 \mid h(i)=j\}$.  Then, expand $T_j$ to $T_j'$
by inserting elements that also map to $j$ under $h$ so that
$|T_1'|=|T_2'|=...=|T_k'|=b$. Clearly, $b\leq F_0$. Therefore, the approximation ratio can be rewritten as \[ F_{g,h}(x)=x^tH'x/x^tx,\]
where $H'_{i,j}=g(i)g(j)I(h(i)=h(j))I(i,j\in \cup_{j=1}^k
T_y')$.

The main technical
lemma needed for the proof of Theorem~\ref{thm:main} is the following, which
essentially says that all points inside any small ball are ``bundled'' together.

\begin{lemma}
\label{probability lemma}
For any $a\in F_{g,h}^{-1}([1-\frac{\epsilon}{2},1+\frac{\epsilon}{2}])\cap P$,
$\Pr[|F(x)-1|\le \epsilon \textrm{ for all } x\in B(a,r) \cap
P]\ge\frac{2}{3}$, where $r=\Omega(\frac{\|a\|_1\epsilon}{poly(F_0)})$.
\end{lemma}

Given a point $a\in \mathbb{R}^{n}$, hash functions $g,h$, and any $-1 < \eps <1$,
denote by $d_{g,h}(a,\epsilon)$ the minimum 1-norm distance between $a$ and
$(F_{g,h}^{-1}(1+\epsilon)\cup F_{g,h}^{-1}(1-\epsilon))\cap P$.  Thus it
is the minimum 1-norm distance from $a$ to the boundary of ``correct region'' using $g$ and
$h$.  Note that $a$ itself may or may not be inside the ``correct region''.  Before
proving Lemma~\ref{probability lemma}, we first establish the following lower bound on
$d_{g,h}(a,\eps)$.

\begin{lemma}
\label{distance lemma}
For any $a\in F_{g,h}^{-1}([1-\frac{\epsilon}{2},1+\frac{\epsilon}{2}])\cap P, -1< \eps < 1$,
$d_{g,h}(a,\eps)=
\Omega(\frac{\|a\|_1\epsilon}{poly(F_0)})$. 
\end{lemma}
\begin{proof}
Let $x^*\in (F_{g,h}^{-1}(1+\epsilon)\cup F_{g,h}^{-1}(1-\epsilon))\cap P$ such that $d_{g,h}(a,\eps)=\|x^*-a\|_1$.

\begin{align}
  \frac{\epsilon}{2}&<|F_{g,h}(x^*)-F_{g,h}(a)|\nonumber\\
  &=|\frac{x^{*t}H'x^*}{x^{*t}x^*}-\frac{a^tH'a}{a^ta}|\nonumber\\
  &<(max_{c\in[0,1]}\|\nabla_y\frac{y^tH'y}{y^ty}\vert_{y=(1-c)a+cx^*}\|_2)\|x^*-a\|_2\nonumber\\
  &<(max_{c\in[0,1]}\|\frac{2((y^ty)H'y-(y^tH'y)y)}{(y^ty)^2}\vert_{y=(1-c)a+cx^*}\|_2)\|x^*-a\|_2\nonumber\\
  &<(max_{c\in[0,1]}\frac{4\|H'\|_2}{\|y\|_2}\vert_{y=(1-c)a+cx^*})\|x^*-a\|_2\nonumber\\
  &<O(\frac{F_0^{\frac{3}{2}}}{\|a\|_1}d_{g,h}(a,\eps))\nonumber
\end{align}

In the last inequality, we have $\|x^*-a\|_2\leq\|x^*-a\|_1=d_{g,h}(a,\eps)$ and $\|y\|_2\geq\|a\|_2\geq\frac{1}{\sqrt{F_0}}\|a\|_1$. To compute $\|H'\|_2$, decompose $H'$ into $H'=UDU^t$, where
$D=\mathrm{diag}[\underbrace{b,b,...,b}_k,0,...,0]$, and $U=[u_1\,u_2\,...\,u_m]$.  For
$i=1,\dots,k$, we set 
$u_i=$ $\frac{1}{\sqrt{b}}\left(
\begin{array}{c}
g(1)I(1\in S_i')\\
g(2)I(2\in S_i')\\
...\\
g(m)I(m\in S_i')\\
\end{array}
\right)$; note that these $u_i$'s are orthonormal. It implies $\|H'\|_2\leq b\leq F_0$.

\end{proof}

We use $d(a,\eps)$ to denote the random variable of $d_{g,h}(a, \eps)$ when $g$ and $h$
are randomly chosen.  We are now ready to prove Lemma \ref{probability lemma}.

\begin{proof}
(of Lemma \ref{probability lemma}) We first rewrite the probability
\begin{align}
&\Pr(|F(x)-1|\le \epsilon \text{ for all } x\in B(a,r)\cap P)
\nonumber\\ 
=&\Pr(|F(a)-1| \le \epsilon \wedge d(a,\epsilon)\ge r \wedge d(a,-\epsilon)\ge r)
\nonumber\\ 
=&1 - \Pr(|F(a)-1| > \epsilon \vee d(a,\epsilon)<r \vee d(a,-\epsilon)<r) \nonumber \\
=&1 - \Pr((|F(a)-1| \le \epsilon \wedge d(a,\epsilon)<r)\vee
(|F(a)-1|\le \epsilon \wedge d(a,-\epsilon)<r)\vee |F(a)-1|>\epsilon). \nonumber 
\end{align}

Next, consider the event $d(a,\epsilon)\le r\wedge|F(a)-1|<\epsilon$.  By
Lemma~\ref{distance lemma}, this event implies that $\epsilon/2<F(a)-1<\epsilon$. Similarly, the event $d(a,-\epsilon)<r\wedge|F(a)-1|<\epsilon$ implies
$-\epsilon/2>F(a)-1>-\epsilon$. Therefore,
\begin{align*}
&\Pr((|F(a)-1| \le \epsilon \wedge d(a,\epsilon)<r)\vee
(|F(a)-1|\le \epsilon \wedge d(a,-\epsilon)<r)\vee |F(a)-1|>\epsilon) \\
\le&\Pr(\epsilon/2<F(a)-1<\epsilon\vee-\epsilon/2>F(a)-1>-\epsilon\vee|F(a)-1|>\epsilon) \\
=&\Pr(|F(a)-1|>\epsilon/2) \\
\le&1/3,
\end{align*}
where the last inequality follows from the error guarantee of the AMS sketch, when using
$k=c/\eps^2$ counters for an appropriate constant $c$.
\end{proof}

We are now ready to finish off the proof of Theorem~\ref{thm:main}.
\begin{proof}(of Theorem~\ref{thm:main})
  Set $r=\Omega(\frac{\|a\|_1\epsilon}{poly(F_0)})$ as in Lemma~\ref{probability
    lemma}.  We divide the stream into epochs such that all frequency vectors inside
  one epoch are within a ball of radius $r$.  Let $f$ and $f +
  \Delta f$ be respectively the frequency vectors at the start and the end of an epoch.
  It is sufficient to have $\| \Delta f\|_1 \le r$, which means that the
    $\ell_1$-norm of the frequency vector increases by a factor of $1+\frac{\epsilon}{poly(F_0)}$
    every epoch.  This leads to a total of $O\left({F_0 \over \epsilon} \log m\right)$
    epochs.  
    
    Suppose we run $l$ independent copies of the AMS sketch and always return the
    median estimate.  Consider any one epoch. Lemma~\ref{probability lemma} has
    established that any one AMS sketch is good for the entire epoch with probability
    at least $2/3$.  If at any time instance, the median estimate is outside the error
    requirement, then that means at least half of the sketches are not good for the
    epoch, which happens with probability at most $2^{-\Omega(l)}$ by a standard Chernoff
    argument.  Finally, by the union bound, the failure probability of the entire stream
    is $2^{-\Omega(l)} \cdot O\left({F_0 \over \epsilon} \log m\right)$, meaning it is sufficient
    to have $l= O\left(\log \left({F_0 \over \epsilon} \log m\right)\right) = O(\log F_0 +
    \log\log m + \log(1/\eps))$.
\end{proof}

\section{Tracking problem of $F_p$ with $p\in(1,2)$}

Indyk's algorithm works as following. Given $l=\Theta(\frac{1}{\epsilon^2}\log \frac{1}{\delta})$, initialize $nl$ independent $p$-stable distribution random variable $X_i^j$, where $i\in[n]$ and $j\in[l]$. Maintain the vector $y=Ax$, where $x$ is the frequency vector and $A_{j,i}=X_i^j$. For query, output $s$-quantile of $|y_j|$ for some suitable $s$. This estimator returns $\epsilon$-approximation with error probability $\delta$. Similar to $F_2$, we have the following theorem.

\begin{theorem}
\label{thm:main2}
  Given a stream $S=(a_1, a_2, ..., a_m)$ where $a_i\in [n]$, let
  $S_i=(a_1,\dots,a_i)$.  If $l=O(\frac{1}{\eps^2}(\log F_0+\log \log m + \log(1/\eps)))$, let $\hat{X}_i$ be the output of the sketches after processing $S_i$, then
  $\Pr\left(\bigwedge_{i=1}^{m}|\hat{X_i}-F_p(S_i)|<\epsilon F_p(S_i)\right)>1/2$.
\end{theorem}

\begin{proof}
Basically, the idea is very similar to the proof for $F_2$ so we point out the main different.

Given $A$, $a\in P$ and any $-1< \eps < 1$, define $F_A(a)=\frac{s-quantile A_ja}{\|a\|_p}$ be the approximation ratio, where $A_j$ is $j$-th row of $A$, and $d_A(a,\eps)$ be the minimum 1-norm distance between $a$ and
$(F_A^{-1}(1+\epsilon)\cup F_A^{-1}(1-\epsilon))\cap P$. Also, denote $x^*\in (F_A^{-1}(1+\epsilon)\cup F_A^{-1}(1-\epsilon))\cap P$ such that $d_A(a,\eps)=\|x^*-a\|_1$.

For fixed $j$, given any $y_1, y_2\in P$ such that $\|y_2\|_1>\|y_1\|_1$,

\begin{align}
  |\frac{A_jy_2}{\|y_2\|_p}-\frac{A_jy_1}{\|y_1\|_p}|&<\sum_{i=1}^n|A_{ji}||\frac{((y_2)_j}{\|y_2\|_p}-\frac{((y_1)_j}{\|y_1\|_p}|\nonumber\\
  &<(max_i|\frac{(y_2)_i}{\|y_2\|_p}-\frac{(y_1)_i}{\|y_1\|_p}|)(\sum_{i=1}^n|A_{ji}|)\nonumber\\
  &<(max_i|(y_1)_j(\frac{\|y_2\|_p-\|y_1\|_p}{\|y_1\|_p\|y_2\|_p})|+|\frac{(y_2-y_1)_j}{\|y_2\|_p}|)(\sum_{i=1}^n|A_{ji}|)\nonumber\\
  &<(max_i|(y_1)_j(\frac{\|y_2-y_1\|_p}{\|y_1\|_p\|y_2\|_p})|+|\frac{(y_2-y_1)_j}{\|y_2\|_p}|)(\sum_{i=1}^n|A_{ji}|)\nonumber\\
  &<(\frac{2\|y_2-y_1\|_1}{\|y_1\|_p})(\sum_{i=1}^n|A_{ji}|)\nonumber
\end{align}
Here, the inequality $\|y_2\|_p-\|y_1\|_p\leq\|y_2-y_1\|_p$ holds when $p\in(1,2)$.

Suppose $a\in F_A^{-1}([1-\frac{\epsilon}{2},1+\frac{\epsilon}{2}])\cap P$, consider the line segment between $a$ and $x^*$, let $a_0=a,a_1,...,a_q=x^*$ be the "switching" point when $s$-quantile is switched in between different $j$.

\begin{align}
  \frac{\eps}{2}&<\sum_{k=0}^{q-1}(\frac{2\|a_{k+1}-a_k\|_1}{\|a_k\|_p})(\sum_{i=1}^n|A_{ji}|)\nonumber\\
  &<\sum_{k=0}^{q-1}(\frac{2\|a_{k+1}-a_k\|_1}{\|a\|_p})(\sum_{i=1}^n|A_{ji}|)\nonumber\\
  &<O((\frac{\|x^*-a\|_1}{\|a\|_p})(\sum_{i,j}|A_{ji}|))\nonumber\\
  &<O(\frac{poly(F_0,l)}{\|a\|_1}d_A(a,\eps))\nonumber
\end{align}

In second last inequality, grouping all the terms with same $j$. In the last inequality, we have $\|a\|_p\geq\frac{1}{F_0^{1-\frac{1}{p}}}\|a\|_1$. For the term $\sum_{i,j}|A_{ji}|$, as they are independent $p$-stable distribution random variable, $\sum_{i,j}|a_{ij}|<C(F_0l)^{\frac{1}{p}}$ for some large constant $C$ with constant probability. Hence, we can conclude that $d_A(a,\eps)=\Omega(\frac{\eps\|a\|_1}{poly(F_0,l)})$

Finally, decompose the stream as in the proof of Theorem~\ref{thm:main}. The total number of epochs is $O\left({poly(F_0,l) \over \epsilon} \log m\right)$. The error probability for each epoch is at most $2^{-\Omega(\eps^2l)}$. Therefore, by taking $l=O(\frac{1}{\eps^2}(\log F_0+\log \log m + \log(1/\eps)))$, the final error probability is $\Theta(1)$.

\end{proof}

Remarks. There are two $F_p$ algorithm in \cite{kane10:_exact} which is more complicated. Our technique may also applied to these algorithms while we left it as future work.

\section{Communication complexity}
We first review the definition of the Augmented-Indexing problem $\Ind(k,N)$. In this problem, Alice has $a\in[k]^N$, and Bob has $t\in[N]$, $a_1\cdots,a_{t-1}$ and $q\in[k]$. (We use $b$ to denote the input of Bob). The function $\ind(a,b)$ evaluates to $1$ if $a_t=q$, and otherwise it evaluates to $0$. The input distribution $\nu$ of the problem defined as follows. $a$ is a uniformly random vector, and $t\in_R[N]$. Set $q=a_t$ with $1/2$ probability and set $q$ randomly with probability $1/2$.   

We define the following communication game, and assume $N\ge 100k$. We have $k+1$ players 
$\{Q, P_1\cdots,P_k\}$. 
Player $Q$ gets a vector $x\in [k]^N$. Let $v\in [N]^k$ be a vector of $k$ distinct indices and $y\in[k]^k$. Each player $P_i$ gets $(v_i,y_i)$, and also a set of pairs $\{(v_j,y_j)~|~v_j > v_i\}$ and a prefix of $x$, i.e. $x_{1:v_i-1}$. $P_i$ needs to decide whether $x_{v_i}=y_i$. Further more, all the players have to answer correctly simultaneously. The communication is one-way, i.e., only player $Q$ sends a message to each of the other players. We use $\Indk$ to denote this communication problem, and we will show that the it has communication complexity $\Omega(k N\log k)$. 

\begin{lemma}
Let $\Pi$ be private coin randomized protocol for $\Indk$ with error probability at most $\delta\le 1/2000$ for any input, then the communication complexity of $\Pi$ is $\Omega(kN\log k)$. 
\end{lemma}
\begin{proof}
We define the input distribution $\mu$ as follows. Pick $x$ uniformly randomly, and the distribution of $v$ is uniform conditioned on all entries in $v$ are distinct. Then for each $i$, with $1/2$ probability, set $y_i=x_{v_i}$ and with $1/2$ probability pick $y_i$ randomly. We will use capital letters to denote corresponding random variables. Let $\{M_1,\cdots, M_k\}$ be the set of messages $Q$ sends to each player respectively. Given the input is sampled from $\mu$, we will show that $H(M_i)=\Omega(N\log k)$ for at least a constant fraction of these messages. In the rest of the proof, the probability is over the random coins in $\Pi$ and the input distribution.

The proof follows the framework of \cite{molinaro2013:_beating}. In our communication problem, each player will know more information about $x$ than Bob in the Augmented-Indexing problem, which introduce more complication. 

Let $L_i=\{j|v_j>v_i\}$, and $E_i$ be the event that $P_i$ answer correctly. We define a set of events $F_i=\{E_j| j\in L_i\}$. Given $V=v$, we can apply the chain rule
$$\Pr(E_1,\cdots, E_k|V=v) = \Pi_i \Pr(E_i|F_i,V=v).$$

By our assumption, $\Pr(E_1,\cdots, E_k|V=v)\ge 1-\delta$. Using the bound $p\le e^{-(1-p)}$ (valid for all $p\in[0,1]$), we have 
$$\sum_i (\Pr(E_i|F_i,V=v)-1)\ge \ln(1-\delta)\ge -10\delta,$$
where the last inequality uses the first-order approximation of $\ln$ at $1$.
Multiplying both side of the above inequality by $\Pr(V=v)$ and then sum over all possible $v$, we have
$$\sum_i(\Pr(E_i|F_i)-1)\ge -10\delta.$$
By Markov's inequality, for at least half of the indices $i$ we have 
$$\Pr(E_i|F_i)\ge 1- 20\delta/k.$$
We call such indices $good$.

Next we give a reduction, using $\Pi$ to solve Augmented-Indexing problem. We hardwire a good index $i$, and call this protocol $\Pi_i$. In this protocol, Bob will simulate the behavior of $P_i$ and Alice simulate the rest of the players. Given an input of the Augmented-Indexing problem $a$ and $b=(t,a_{<t},q)$, which is sampled from $\nu$, Alice sets $x=a$, and then samples $k-1$ distinct indices $v_{-i}$ and corresponding $y_{-i}$ according to our input distribution $\mu$ (here we use $v_{-i}$ to denote the vector $v$ excluding the $i$th coordinate) and send $v_{-i}$ and $y_{-i}$ to Bob. Bob sets $v_i=t$, $y_i=q$ and $x_{<v_i}=a_{<t}$. Bob then checks whether there is some $j\neq i$ such that $v_j=t$, and if there is, Bob output 'abort'. Notice this only happens with probability $1/100$ since we assume $N\ge 100k$. It is easy to verify that, conditioned on this not happening, the input we constructed for $\Pi$ is exactly the same as $\mu$.

Alice then runs $\Pi$, simulating player $Q$, and sends $M_i$ to Bob. Also Alice computes the answer of $P_j$ for $j\neq i$ based on the message $M_j$. Note that Alice will get the same answer as $P_j$ for $j\in L_i$, since Alice has the entire input of such $P_j$. Let $o_j$ be the answer of $P_j$. Alice then check whether each of the answers is correct, and finds the largest $v_j$ such that $o_j$ is not correct and send $s=v_j$ to Bob (if there exists one, and otherwise set $s=0$). 

Bob check whether $s>t$, and if so, output 'abort', and this happens with probability at most $\delta$.  Bob then runs $\Pi$ simulating player $P_i$, and outputs whatever $P_i$ outputs, because Bob has $v$ and $y$. 

Notice that $o_j$ is correct for all $j\in L_i$ if and only if $s<t$, and in this case all the events in $F_i$ happen. So the probability that the above protocol outputs 'abort' is at most $1/50$. Conditioned on 'abort' does not happen, we have $s< t$, which implies that all events in $F_i$ happen, and the success probability of $\Pi_i$ is at least $\Pr(E_i|F_i)\ge1-20\delta/k$, since $i$ is good.

We next analyze the information cost of the protocol $\Pi_i$ for a good $i$.
By definition
$$I(M_i,S;X|V,Y)=H(X|V,Y)-H(X|M_i,V,Y,S).$$

It is easy to see $H(X|V,Y)\ge \frac{99N}{100}\log k$, since we assume $N\ge 100k$. So we only need to upper bound $H(X|M_i,V,Y,S)$. In $\Pi_i$, the message send by Alice is $M_i,V_{-i},Y_{-i},S$, then Bob 'abort' with probability at most $1/50$, and condition on not 'abort', Bob outputs a correct answer with probability at least $1-1/100k$. With such properties, we have the following lemma, which is shown in \cite{molinaro2013:_beating}.
\begin{lemma}\label{Fano}
$H(X|M_i,V_{-i},Y_{-i},S)\le\frac{1}{20}N\log k$. 
\end{lemma}
For completeness, we provide a proof of the above lemma in Appendix \ref{app:fano}. By property of entropy, $H(X|M_i,V,Y,S)\le H(X|M_i,V_{-1},Y_{-i},S)$, so we have

$$I(M_i,S;X|V,Y)=\Omega(N\log k).$$
Because $S$ only takes $\log N$ bits, we get
$$H(M_i)\ge H(M_i,S)-\log N\ge I(M_i,S;X|V,Y)-\log N=\Omega(N\log k).$$
We have shown that there are at least $k/2$ good $i$, so we prove that the communication complexity of $\Pi$ is $\Omega(kN\log k)$.
\end{proof}

\section{Lower bound of tracking $F_p$ in the cash register model}
We will give a reduction from $\Indk$. For convenience, we change the definition of $\Indk$ slightly. Here each player $P_j$ gets $(v_j,y_j)$, and also a set of pairs $\{(v_{\ell},y_{\ell})~|~v_{\ell} < v_j\}$ and a suffix of $x$, i.e. $x_{v_j+1:N}$. $P_j$ needs to decide whether $x_{v_j}=y_j$.
Clear this new problem is equivalent to $\Indk$.  
For $0<p\le 2$ and $p\neq 1$, we define $t=\frac{ 2^{2p}}{|2^{p-1}-1|}$ and $q=t^{1/p}$. We have the following lower bound for tracking $F_p$.
\begin{theorem}
For any linear sketch based algorithm which can track $F_p$ continuously within accuracy
$(1\pm \frac{2^p-2}{2^{p+3}})$ in the cash register model, the space used is at least $\Omega(\log m\log\log m/\log q)$ bits.
\end{theorem}
\begin{proof}
Given an linear sketch algorithm which can track $F_2$ within error $(1\pm\eps)$ of an incremental stream at all time with probability $1-\delta$, we show how to use this algorithm to solve $\Indk$ defined above. We use $L$ to denote the algorithm, and $L(f)$ to denote the memory of the algorithm when the input frequency vector is $f$. For linear sketch, the current state of algorithm does not depend on the order of the stream, and only depends on the current frequency vector. Let $O(L(f))$ denote the output of the algorithm, when the current memory state is $L(f)$.

Given $x$, $Q$ runs the following reduction, which use similar ideas as in \cite{kane10:_exact}. Each item in the stream is a pair $(i,x_i)$. For $i\in[N]$, $Q$ insert $\lfloor q^{i}\rfloor$ items $(i,x_i)$. We use $f(x)$ to denote the frequency vector of this stream, and we can also view $f:[k]^N\rightarrow \integer^{Nk}$ as a linear transformation. Then $Q$ runs the streaming algorithm to process the stream just constructed, and sends the memory content $L(f(x))$ to each of the other players. 

For each $j$, $P_j$ first computes $L(f(x_{\le v_j}))=L(f(x-x_{>v_j}))$. Since the $L$ and $f$ are linear, we can do this. Then $P_j$ inserts $\lfloor q^{v_{\ell}}\rfloor$ copies of $(v_{\ell},y_{\ell})$ for all $\ell$ such that $v_{\ell}< v_j$. $P_j$ can do this because he knows all $(v_{\ell},y_{\ell})$ with $v_{\ell}<v_j$. Now the frequency vector is 
$$u=f(x_{\le v_j}+\sum_{\ell:v_{\ell}<v_j}y_{\ell}\cdot e_{v_{\ell}}),$$
where we use $e_j$ to denote the $j$th standard basis vector. Note that $\lfloor q^i\rfloor ^p = t^i(\lfloor q^i\rfloor/q^i)^p$, and thus $t^i/2^p\le \lfloor q^i\rfloor^p\le t^i$.
 Let $F_p(f)$ be the $p$th moment of $f$. We first consider the case when $p>1$.
 We have
 $$F_p(f(x_{<v_j}))+\lfloor q^{v_j}\rfloor^p\le F_p(u) \le 2^p\cdot F_p(f(x_{<v_j}))+\lfloor q^{v_j}\rfloor^p .$$ 
Then $P_j$ inserts 
$\lfloor q^{v_{j}}\rfloor$ copies of $(v_j,y_j)$, making the frequency vector 
$$f(x_{\le v_j}+\sum_{\ell:v_{\ell}\le v_j}y_{\ell}\cdot e_{v_{\ell}}).$$ 
We use $f_j$ to denote this vector.

When $y_j=x_{v_j}$, then $F_p(f_j)$ is at least $F_p(f(x_{<v_j}))+2^p\lfloor q^{v_j}\rfloor^p$.
On the other hand, if $y_j\neq x_{v_j}$, then $F_p(f_j)$ is at most $2^p\cdot F_p(f(x_{<v_j}))+2 \lfloor q^{v_j}\rfloor^p$. The difference between these two is at least 
$$(2^p-2)\cdot \lfloor q^{v_j}\rfloor^p-2^p\cdot F_p(f(x_{<v_j}))\ge \frac{2^p-2}{2^p}\cdot t^{v_j} -2^p\cdot \frac{t^{v_j}}{t-1},$$ 
which is at least $\frac{2^p-2}{2^{p+3}}$ fraction of the current $F_p$ by our setting of $t$. So if the algorithm can estimate $F_p$ of $f_j$ within accuracy $(1\pm \frac{2^p-2}{2^{p+3}})$, then $P_j$ can distinguish these two cases.

Now we need to argue that all the players can output a correct answer simultaneously. More precisely, we need $O(L(f_j))$ to be correct simultaneously for all $j$. In order to prove this, we construct an incremental stream such that all $f_j$ will appear at sometime during the stream.

Let us consider the following stream $S$. Let $\P$ be a permutation such that $\P(v)$ is the sorted and let $v'=\P(v)$ and  $y'=\P(y)$. 
The stream has $k$ phases. In phase $j$ for $1\le j\le k$, we inserts $\lfloor q^{i}\rfloor$ copies of $(i,x_i)$ for all $v'_{j-1}<i\le v'_{j}$, and inserts $\lfloor q^{v'_{j}}\rfloor$ copies of $(v'_{j},y'_{j})$. (Here we set $v'_{0}=0$). So the number of items inserted in this stream is at most $2 \cdot q^{N}$. Let $S_{\tau}$ be the stream after $\tau$th phase ends and $s_\tau$ be the corresponding frequency vector of $S_\tau$.

We can verify that, for every $j$, $f_j=s_\tau$, where $\tau$ is the rank of $v_j$ in $v$. Since the algorithm is linear sketch, the output only depends on the frequency vector of the current stream. So the output of $P_j$ is correct, as long as the output of the streaming algorithm is correct at time $\tau$. However by our assumption, the streaming algorithm will succeed at all time, since $S$ is incremental. So
all the players will give correct answers simultaneously with probability $(1-\delta)$, which solves the communication game.

Now suppose the streaming algorithm use space $T$, then the communication cost of the above protocol is $kT$. So $T=\Omega(N\log k)$. The length of $S$ is at most $2\cdot q^N$. Given $m$, we set $N=\log_q \frac{m}{4}/2$  and $k=\log_q \frac{m}{4}/200$, so the length of $S$ is at most $m$, and $T=\Omega(\log m\log\log m/\log q)$ for fixed $1<p\le 2$. The case when $0<p<1$ is similar. 

\end{proof}

\section{Lower bound for tracking $F_p$ in turnstile model}
In the turnstile model, negative updates are allowed in the stream, so the communication game we reduce from is simpler. The problem we will use is the $k$-fold version of $\Ind(k,N)$. More precisely, Alice has $\{a^1,\cdots,a^k\}$ and Bob has $\{b^1,\cdots,b^k\}$, where each pair $(a^i,b^i)$ is an input for $\Ind(k,N)$ and $b^i=(t_i, a^i_{>t_i},q_i)$, and they want to compute a $k$-bit vector $o$, such that $o_i=\ind(a^i,b^i)$ for all $i\in[k]$. We call this problem $\Ind^k(k,N)$, and we have the following results from \cite{molinaro2013:_beating}.
\begin{theorem}
For any integer $k$ and $N$, the communication complexity of solving $\Ind^k(k,N)$ with constant probability is $\Omega(kN\log k)$.
\end{theorem}
\begin{theorem}
For any linear sketch based algorithm which can track $F_p$ continuously within accuracy
$(1\pm 0.5)$ in the turnstile model, the space used is at least $\Omega(p\log^2 m)$ bits.
\end{theorem}
\begin{proof}
For $0<p\le 2$, we define $q=2^{1/p}$.
Let $f:[k]^N\rightarrow \integer^{Nk}$ be the same linear function defined as above. We next give a reduction from $\Ind^k(k,N)$ to tracking $F_p$ in the turnstile model. Let $L$ be a linear sketch. Alice sends $L(f(a^i))$ for $i\in[k]$ to Bob. For each $i$, Bob compute a sketch $\Gamma_i=L(f(a^i-a^i_{>t_i}-q_ie_{t_i}))$, where $e_j$ is the $j$th vector in the standard basis. It is easy to verify that $2^{t_i}/2^p \le F_p(f(a^i_{<t_i}))\le 2^{t_i}$. When $q_i=a^i_{t_i}$, we have $f(a^i-a^i_{>t_i}-q_ie_{t_i})=f(a^i_{<t_i})$, and $F_p(f(a^i_{<t_i}))$ is at most $2^{t_i}$. On the other hand, if $q_i\neq a^i_{t_i}$, then $F_p(f(a^i-a^i_{>t_i}-q_ie_{t_i}))\ge 2\cdot 2^{t_i}$. So the if the output $O(\Gamma_i)$ is within accuracy $(1\pm 0.5)F_p$, then Bob can distinguish these two cases, and output $\ind(a^i,b_i)$ correctly. 

Now we need to prove that Bob can solve $k$ instances simultaneously. As in the proof in cash register model, we construct a imaginary stream such that the frequency vector each $\Gamma_i$ sketched appear in the stream at sometime, but now we can use negative updates in the stream.

The stream has $k$ phases, and the $i$th phase corresponds to $(a^i,b^i)$. In the $i$th phase, we first inserts sets of positive updates, so that at the end the frequency vector is $f(a^i_{\le t_i})$, then inserts $\lfloor q^{t_i} \rfloor$ copies of $((t_i,q_i),-1)$. The last step in this phase is to reverse all the above updates, so that the frequency vector becomes $0$. Clearly, for each $i$, the frequency vector before the cleaning step in phase $i$ is exactly $f(a^i-a^i_{>t_i}-q_ie_{t_i})$, which is the vector sketched by $\Gamma_i$. We call the above stream $S$. By our assumption, the algorithm $L$ is correct at any time during the stream $S$, so Bob solves $\Ind^k(k,N)$. The communication cost is $kT$, where $T$ is the amount of space used by $L$, and thus $T=\Omega(N\log k)$. The number of updates in $S$ is at most $k 2^{N/p}$. Given $m$, we set $k=\sqrt{m}$, and $N=p\log \frac{m}{2}$, so that the number of update is bounded by $m$, and we have $T=\Omega(p\log^2 m)$. 

\end{proof}

\bibliography{paper}
\bibliographystyle{abbrv}

\appendix 
\section{Proof of Lemma \ref{Fano} }\label{app:fano}
\begin{proof}
We focus on tuples $(t,a,r)$, where $r$ is the private coins used by Alice and Bob in $\Pi_i$, including the randomness used in $\Pi$ and $v_{-i},y_{-i}$ which are sampled by Alice. Let $$U_1=\{(t,a,r)~:~\Pi_i(a,t,a_t,r)='abort'\}.$$
Here we use $\Pi_i(a,t,a_t,r)$ to denote the output of $\Pi_i$ with input $a,t,a_t$  (yes instance) and random coins $r$.
We use $f(a,t,q)$ to denote the corresponding function of Augmented-indexing problem.
We define 
$$U_2=\{(t,a,r):\exists q \text { st }\Pi_i(a,t,q,r)\neq f(a,t,q) \wedge \Pi_i(a,t,q,r) \neq 'abort'\}.$$
We say a tuple $good$ if it does not belong to either $U_1$ or $U_2$. Notice that if 
$(t,a,r)$ is good, then: (1) $\Pi_i(a,t,a_t,r)=1$; (2) for every $q\neq a_t$, $\Pi_i(a,t,q,r)\neq 1$.
\begin{lemma}
For every index $t\in{N}$, there is a predictor $g_t$ such that 
$$\Pr(g_t(M_i(A),A_{<t}, V_{-i},Y_{-i}, S )=A_t)\ge \Pr((t,A,R)\text{ is good}).$$
\end{lemma}
\begin{proof}
We set $g'_t(M_i(a),a_{<t}, v_{-i},y_{-i}, s,r_B )$ to any value $q$ such that $\Pi_i(a,t,q,r)=1$, where $r_B$ is the random coins used by Bob. (if no such $q$ exists, set arbitrarily). By the above argument, if $(t,a,r)$ is good, $g'_t(M_i(a),a_{<t}, v_{-i},y_{-i}, s, r_B )=a_t$, which shows that $\Pr(g'_t(M_i(A),A_{<t}, V_{-i},Y_{-i}, S ,R_B)=A_t)\ge \Pr((t,A,R)\text{ is good})$. By definition, we have
\begin{eqnarray*}
&&\sum_{r_B}\Pr(R_B=r_B)\Pr\left((T,A,R)\text{ is good }|R_B=r_b\right)\\
&=&\Pr((T,A,R)\text{ is good }),
\end{eqnarray*}
so there is an $r_b$, such that $\Pr\left((T,A,R)\text{ is good }|R_B=r_b\right)\ge
\Pr((T,A,R)\text{ is good })$.
We then set $g_t(M_i(A),A_{<t}, V_{-i},Y_{-i}, S )=g'_t(M_i(A),A_{<t}, V_{-i},Y_{-i}, S ,r_b)$
 which proves the lemma.
\end{proof}

Then by Fano's inequality, we have that
$$H(A_t|M_i(A),A_{<t}, V_{-i},Y_{-i}, S) \le 1+\log k\cdot(\Pr((t,A,V_{-i},Y_{-i})\text{ is not good})).$$

By definition
\begin{eqnarray*}
H(A|M_i(A),V_{-i},V_{-i},S)&=&\sum_{t=1}^N H(A_t|M_i(A),A_{<t}, V_{-i},Y_{-i}, s)\\
&\le& N+\log k\sum_{t=1}^N\Pr((t,A,V_{-i},Y_{-i})\text{ is not good})
\end{eqnarray*}

\begin{lemma}
$Pr((T,A,R)\text{ is not good})\le 1/20$.
\end{lemma}
\begin{proof}
By union bound, we only need to show that the probability 
$$Pr((T,A,R)\in U_1)+Pr((T,A,R)\in U_2)\le 9/20.$$
We have 
\begin{eqnarray*}
Pr((T,A,R)\in U_1) &=& \Pr(\Pi_i(A,T,A_T,R)=abort)\\
&=&\Pr(\Pi_i(A,T,Q,R)=abort|Q=A_T)\\
&=& \Pr(\text{protocol aborts}|Q=A_T).
\end{eqnarray*}
Since $\Pr(Q=A_T)=1/2$ and  $\Pr(\text{protocol aborts})\le 1/50$, we have $\Pr((T,A,R)\in U_1)\le 1/25$.
We also have
\begin{eqnarray*}
\Pr((T,A,R)\in U_2) &=&\Pr\left[\vee_{q\in[k]}(\Pi_i(A,T,q,R)\neq f(A,T,q) \wedge  \Pi_i(A,T,q,R)\neq abort)\right]\\
&\le& \sum_{q\in[k]}\Pr\left[ \Pi_i(A,T,q,R)\neq f(A,T,q) \wedge  \Pi_i(A,T,q,R)\neq abort\right] \\
&\le&\sum_{q\in[k]} \Pr\left[\Pi_i(A,T,q,R)\neq f(A,T,q)|\Pi_i(A,T,q,R)\neq abort\right]\\
&\le& k\cdot \Pr\left[\Pi_i(A,T,Q,R)\neq f(A,T,Q)|\Pi_i(A,T,Q,R)\neq abort\right]\\
&\le& k \cdot \frac{1}{100k} =1/100.
\end{eqnarray*}
So $Pr((T,A,R)\text{ is not good})\le 1/20$.
\end{proof}
As the distribution of $T$ is uniform, we have
$$\sum_{t=1}^N\Pr((t,A,V_{-i},Y_{-i})\text{ is not good})=N\cdot \Pr((T,A,V_{-i},Y_{-i})\text{ is not good})\le 9N/20$$

\end{proof}

\end{document}